\documentclass[11pt]{article}
\pdfoutput=1

\usepackage{amsfonts}
\usepackage{amssymb}
\usepackage{amsmath}
\usepackage{amsthm}
\usepackage{subfigure}
\usepackage{psfrag}
\usepackage{url, fullpage, color, upgreek}
\usepackage{graphicx}

\newcommand{\Rbb}{\mathbb{R}}
\newcommand{\scp}[2]{\langle #1, #2 \rangle}
\newcommand{\Nbb}{\mathbb{N}}
\newtheorem{theorem}{Theorem}%[section]

\newtheorem{lemma}{Lemma}

\newcommand{\inv}[1]{\frac{1}{#1}}
\newcommand{\supp}{{\rm supp}\,}
\newcommand{\tinv}[1]{{\textstyle\frac{1}{#1}}}
\newcommand{\sign}{{\rm sign}\,}

\DeclareMathOperator*{\argmin}{argmin}

\DeclareMathOperator{\prox}{prox}
\newcommand{\Proj}{\mathcal{P}}
\DeclareMathOperator{\Id}{Id}
\DeclareMathOperator{\dom}{dom}

\newcommand{\ie}{\mbox{i.e.~}}
\newcommand{\norm}[1]{\|#1\|}
\newcommand{\Hm}{\mathcal{H}}
\newcommand{\iBPDN}{\textrm{\mbox{\emph{i}\hspace{.5pt}BPDN}}}
\newcommand{\Ncal}{\mathcal{N}}
\newcommand{\cl}{\mathcal}
\newcommand{\wt}{\widetilde}
\title{A Short Note on Compressed Sensing\\ with Partially Known
  Signal Support\\[3.5mm]
\footnotesize Technical Report: TR-LJ-2009.02 (Updated Version)} \date{\today}
\author{L. Jacques\footnote{{\tt
laurent.jacques@uclouvain.be}. Research supported by Belgian National Science Foundation (F.R.S.-FNRS).}\\[2mm]
  \footnotesize Communications and Remote Sensing Laboratory (TELE)\\
  \footnotesize Universit\'e catholique de Louvain (UCL)\\
  \footnotesize Louvain-la-Neuve, Belgium}

\begin{document}

\maketitle

\begin{abstract}
This short note studies a variation of the Compressed Sensing
paradigm introduced recently by Vaswani et al., \ie the recovery of
sparse signals from a certain number of linear measurements when the
signal support is partially known. The reconstruction method is
based on a convex minimization program coined \emph{innovative}
Basis Pursuit DeNoise (or \iBPDN). Under the common
$\ell_2$-fidelity constraint made on the available measurements,
this optimization promotes the ($\ell_1$) sparsity of the candidate
signal over the complement of this known part.

In particular, this paper extends the results of Vaswani et al. to the
cases of compressible signals and noisy measurements. Our proof relies
on a small adaption of the results of Candes in 2008 for
characterizing the stability of the Basis Pursuit DeNoise (BPDN) program.

We emphasize also an interesting link between our method and the
recent work of Davenport et al. on the $\delta$-stable embeddings and
the \emph{cancel-then-recover} strategy applied to our problem. For
both approaches, reconstructions are indeed stabilized when the
sensing matrix respects the Restricted Isometry Property for the same
sparsity order.

We conclude by sketching an easy numerical method relying on monotone
operator splitting and proximal methods that iteratively solves \iBPDN.
\end{abstract}

\begin{quote} 
\end{quote}

\section{Introduction}
\label{sec:intro}

The theory of Compressed Sensing (CS)
\cite{candes2006qru,donoho2006cs} aims at reconstructing sparse or
compressible signals from a small number of linear measurements
compared to the dimensionality of the signal space.  In short, the
signal reconstruction is possible if the underlying sensing matrix is
well behaved, \ie if it respects a Restricted Isometry Property (RIP)
saying roughly that any small subset of its columns is ``close'' to an
orthogonal basis. The signal recovery is then obtained using
non-linear techniques based on convex optimization promoting signal
sparsity, as the Basis Pursuit DeNoise (BPDN) program
\cite{donoho2006cs,Chen98atomic}. What makes CS more than merely an
interesting theoretical concept is that some classes of randomly
generated matrices (e.g. Gaussian, Bernoulli, partial Fourier
ensemble, etc) satisfy the RIP with overwhelming probability. This
happens as soon as their number of rows, \ie the number of CS
measurements, is higher than a few multiples of the assumed signal
sparsity.

In this paper we are interested in a variation of the CS paradigm. We
assume indeed that the support of the signal to recover is partially
known, possibly with a certain error.  As explained in
\cite{vaswani5066modified,Vaswani2009}, this context is indeed well
suited to the recovery of (time) sequences of sparse signals when
their supports evolves slowly over time. In that case, the support of
the recovered signal in a previous (discretized) time can be used to
improve the reconstruction of the signal at the next time instance,
either by decreasing the required number of measurements for a given
quality, or by improving the reconstruction quality for a fixed number
of measurements.  Recovering a signal with partially known support is
also of interest for certain kind of 1-D signals or images. For
instance, photographic images, \ie with positive intensities, have
often many non-zero approximation coefficients in their wavelet
decomposition \cite{mall99}; a prior knowledge that can be favorably
used in their reconstruction from CS measurements.

By adapting the proof of \cite{candes2008rip}, we show in this short
note that the recovery algorithm minimizing the $\ell_1$-norm of the
signal candidate over the complement of the known support part, \ie
what we coin \emph{innovative} Basis Pursuit DeNoising (\iBPDN), has a
similar stability behavior than the common Basis Pursuit DeNoise
program. In particular, this extends the result of
\cite{vaswani5066modified,Vaswani2009} to the cases of noisy
measurements and of compressible signals, \ie with non-zero but fast
decaying coefficients in a given sparsity basis. We show also that our
method shares somehow the conclusion of the \emph{cancel-then-recover}
strategy designed in \cite{Davenport_M_2010_article_sigproccompress}
where Authors propose a recovery algorithm that applies an orthogonal
projection to separate the measurements into two components, and then
recovers the known support part of the signal separately from the
unknown support component.

\section{Framework and Notations}
\label{sec:fmwk}

Let $x=\Psi\alpha\in\Rbb^n$ be a sparse or a compressible discrete
signal in the sparsity basis $\Psi\in\Rbb^{n\times n}$ of $\Rbb^n$,
\ie the vector $\alpha\in\Rbb^n$ has few non-zero or fast decaying
components respectively. For the sake of simplicity, we work hereafter
with the canonical basis, \ie $\Psi=\Id$, identifying $\alpha$ with
$x$. The present work is however valid for any orthonormal $\Psi$,
e.g. the DCT or the Wavelet basis, by integrating $\Psi$ in the
sensing model described in Section \ref{sec:sensing-model}.

We now establish some important notations. We write
$\Ncal=\{1,\,\cdots,n\}$ the index set of the vector components in
$\Rbb^n$.  For any vector $u\in\Rbb^n$, $u_i$ is the $i^{\rm th}$
component of $u$ with $i\in\mathcal{N}$, $u_S$ is the vector equal to
the components of $u$ on the set $S\subset\Ncal$ and to 0 elsewhere,
while $u^l$, with uppercase index $l\in\Nbb$ to avoid confusion, is
the vector obtained by zeroing all but the $l$ largest components of
$u$ (in amplitude). For non-trivial basis $\Psi$, $u^l$ would be the
best $l$-term approximation of $u$ in the $\ell_2$-norm sense. The
complement of any set $S\subset\Ncal$ is denoted by $S^c =
\Ncal\setminus S$, and the size of $S$ by $\#S$. The $\ell_p$ norm
(for $p\geq 1$) of $u\in\Rbb^n$ is $\|u\|_p^p=\sum_i |u_i|^p$, while
its support is written $\supp u \triangleq \{i\in\Ncal:u_i\neq
0\}$. By extension, the $\ell_0$ ``norm''\footnote{It is not actually
a true norm since for instance it is not positive homogeneous.} is defined as
$\|u\|_0=\#\,\supp u$.

Let us speak now of the prior knowledge that we have on the signal. In
addition to the assumption of sparsity or compressibility, we presume
that the support of the signal $x$ is partially known. In the sequel,
we denote the known support part by $T\subset\Ncal$, while we always
refer to its size by the letter $s=\#T$. Notice that in our study
nothing prevents $T$ to be corrupted by some ``noise'', \ie a priori
$T$ is not fully included to $\supp x$. Moreover, the size of $(\supp
x) \setminus T$ is not constrained, what will matter is the values of
the components of $x$ on $(\supp x) \setminus T$, \ie the
compressibility of $x$ outside of $T$.

\section{Sensing Model}
\label{sec:sensing-model}

Following the common Compressed Sensing model, our vector $x$ is
acquired by a sensing matrix $\Phi\in\Rbb^{m\times n}$ subject to an
additional white noise $n\in\Rbb^m$, \ie
$$
y\ =\ \Phi x + n, 
$$
where $y\in\Rbb^m$ is the measurement vector. In this model the noise
power is assumed bounded\footnote{Possibly with high probability.} by
$\epsilon$, $\|n\|_2\leq \epsilon$.

As shown after, even if a part of the signal support is known, the
stability of this sensing model, \ie our ability to recover or
approximate $x$ from $y$, is also linked to the \emph{Restricted Isometry
Property} (RIP) of the sensing matrix
\cite{candes2005dlp,candes2006qru,candes2006ssr}. 

Explicitly, the matrix $\Phi\in\Rbb^{m\times n}$ satisfies the RIP of order
$q\in\Nbb$ ($q\leq n$) and radius $0 \leq \delta_q < 1$, if
$$
(1-\delta_q)\|u\|_2^2\ \leq\ \|\Phi u\|_2^2\ \leq\ (1+\delta_q)\|u\|^2_2,
$$
for all $q$-sparse vectors $u\in\Rbb^n$, \ie with $\|u\|_0 \leq q$.

\section{Reconstructing on Innovation}
\label{sec:rec-meth}

Intuitively, if a part $T\subset\Ncal$ of the signal support is known,
a possible (non-linear) reconstruction technique of $x$ would simply
consist in minimizing the sparsity of a signal candidate $u\in\Rbb^n$
over $T^c$, \ie the $\ell_0$-norm of $u_{T^c}$, subject to the common
$\ell_2$ fidelity constraint $\|\Phi u - y\|_2\leq \epsilon$ as
prescribed by the noise power bound. As underlined many times in the
community, such a procedure would result in a combinatorial (NP-hard)
problem \cite{natarajan1995sas}. Here again an $\ell_1$
\emph{relaxation} must be used, with possibly additional requirements
on the RIP-``conditioning'' of $\Phi$ \cite{Tropp2006,candes2006ssr}.

The proposed method is a simple extension of the Modified-CS scheme
defined in \cite{vaswani5066modified,Vaswani2009}. We integrate indeed
the case of corrupted measurements by defining the following
optimization program, coined \emph{innovative} Basis Pursuit DeNoising
(\iBPDN),
\begin{equation*}
  \label{eq:IBPDN}
  \argmin_u \norm{u_{T^c}}_1\ {\rm s.t.}\ \norm{y - \Phi u}_2\leq
  \epsilon. \eqno{({\bf \iBPDN})}
\end{equation*}
The term ``innovative'' recalls that this program tries to minimize
the sparsity of the signal to be reconstructed in the unknown (or
innovation) set $(\supp x)\setminus T$ included to $T^c$.

\section{{\em i\hspace{1pt}}BPDN\ and $\ell_2-\ell_1$ Instance Optimality }
\label{sec:ibpdn-ell_2-ell_1}

The main result of this note provides the conditions under which the
solution of \iBPDN\ is close or equal to the initial signal $x$,
\ie the so-called $\ell_2-\ell_1$ instance optimality
\cite{Cohen-bestkterm}. It extends in the same time the conclusion of
\cite{vaswani5066modified,Vaswani2009} to the cases of noisy measurements and
compressible signals.

\begin{theorem}
\label{thm:l2l1-inst}
Under the condition of the sensing model described above, writing $\#T
= s$ and given $k\in\Nbb$, let us assume that the matrix $\Phi$
respects the RIP of order $s+2k$ with radius $\delta_{s+2k}\in (0,1)$, and that
its radius for the smaller order $2k$ is $\delta_{2k}\in (0,1)$. Then, if
  $\delta_{2k}^2 + 2 \delta_{s+2k} < 1$, \iBPDN\ has the
  $\ell_2-\ell_1$ instance optimality meaning that its solution $x^*$
  respects
  $$
\|x-x^*\|_2\ \leq\ C_{s,k}\,\epsilon\ +\ D_{s,k}\,e_{0}(r;k),
$$
where $r$ is the residual $r=x - x_T$, and $e_0(r;k)=k^{-1/2}\|r -
r^k\|_1$ is the compressibility error\footnote{It could be called also
  \emph{scaled} $\ell_1$-\emph{approximation error}.} at $k$-term of
$r$.  The two constants $C_{s,k}$ and $D_{s,k}$, given in the proof,
depend on $\Phi$ only. For instance, for small innovation, \ie when
$k\ll s$, if $\delta_{2k}=0.02$ and if $\delta_{s+2k}=0.2$,
$C_{s,k}<7.32$ and $D_{s,k}<3.35$.
\end{theorem}

\begin{proof}
  We basically adapt the proof of \cite{candes2008rip} to signal
  with partially known support.

  We define the residual $r = x - x_T$, with $\supp r = (\supp x)
  \setminus T$. Let us write $x^*=x+h$ with $h\in\Rbb^n$ so that the
  proof amounts to bound $\|h\|_2$. Let $T_0$ be the support of the
  $k$ largest coefficients of the residual $r = x - x_T$,
  \ie $T_0=\supp r^k$ with $T_0\,\cap\, T = \emptyset$. 

  We define next the sets $T_j$ for $j\geq 1$ as the support of the
  $k$ largest coefficients of $h_{S_j^c}=h-h_{S_j}$ with $S_j=
  T\,\cup\,\bigcup_{l=0}^{j-1}T_l$. By construction, we may observe
  that we got the partition $\bigcup_{l\geq 0} T_l = (\supp
  x)\setminus T$, with $\#T_j=k$ and $T_j\cap T =T_j\cap T_{j'}
  =\emptyset$, for $j,j'\geq 0$ and $j\neq j'$.

  Let us write $T_{|0}=T\cup T_0$ and $T_{|01}=T\cup T_0\cup T_1$,
  with $\#T_{|0}=s+k$ and $\#T_{|01}=s+2k$. The plan of the proof is
  to first bound $\|h_{T_{|01}^c}\|_2$ and then $\|h_{T_{|01}}\|_2$.

  Using the triangular inequality, we have $\|h_{T_{|01}^c}\|_2\leq
  \sum_{j\geq 2} \|h_{T_j}\|_2$. For $j\geq 1$, $\|h_{T_j}\|_1\geq
  k\|h_{T_{j+1}}\|_\infty$ by the ordering of the $T_j$'s, and
  therefore $\|h_{T_{j+1}}\|^2_2\leq k\|h_{T_{j+1}}\|_\infty^2 \leq
  \inv{k}\|h_{T_j}\|^2_1$. This leads to
\begin{equation}
\label{eq:first-h_Tb01-bound}
\|h_{T_{|01}^c}\|_2\ \leq\ \tinv{\sqrt{k}}\,\sum_{j\geq
  1}\|h_{T_j}\|_1\ =\ \tinv{\sqrt{k}}\|h_{T_{|0}^c}\|_1.
\end{equation}
Since $T^c=T_0\cup T_{|0}^c$ and $\|n\|_2=\|y-\Phi
x\|_2\leq\epsilon$, and because $x^*$ solves \iBPDN, we have
$$
\|x_{T^c}\|_1\geq \|x_{T^c} + h_{T^c}\|_1 = \|x_{T_0} + h_{T_0}\|_1\ +\ \|x_{T_{|0}^c} + h_{T_{|0}^c}\|_1
\geq \|x_{T_0}\|_1 - \|h_{T_0}\|_1 - \|x_{T_{|0}^c}\|_1 +
\|h_{T_{|0}^c}\|_1,
$$
and therefore,
$$
\|h_{T_{|0}^c}\|_1 \leq \|x_{T^c}\|_1 +
\|x_{T_{|0}^c}\|_1 + \|h_{T_0}\|_1 - \|x_{T_0}\|_1\\
= 2\|x_{T_{|0}^c}\|_1 + \|h_{T_0}\|_1 = 2\|r - r_{T_0}\|_1 + \|h_{T_0}\|_1.
$$
Consequently, using (\ref{eq:first-h_Tb01-bound}) and the equivalence
of the norms $\ell_2$ and $\ell_1$, we get
\begin{equation}
\label{eq:bound-on-sum-j-gt-2}
\|h_{T_{|01}^c}\|_2\leq
  \sum_{j\geq 2} \|h_{T_j}\|_2 \leq\ 2e_0(r;k) + \|h_{T_0}\|_2.
\end{equation}

Let us now bound $\|h_{T_{|01}}\|_2$. Notice that $h_{T_{|01}} = h -
\sum_{j\geq 2} h_{T_j}$, so that, using Cauchy-Schwarz,
\begin{align*}
  \|\Phi h_{T_{|01}}\|_2^2&=\ \scp{\Phi h_{T_{|01}}}{\Phi h_{T_{|01}}}\\
  &=\ \scp{\Phi h_{T_{|01}}}{\Phi h} - \scp{\Phi h_{T_{|01}}}{\textstyle\sum_{j\geq 2} \Phi h_{T_j}}\\
  &\leq\ \|\Phi h_{T_{|01}}\|_2\|\Phi h\|_2 + \textstyle\sum_{j\geq 2}
  |\scp{\Phi h_{T_{|01}}}{\Phi h_{T_j}}|.
\end{align*}

By hypothesis, $\Phi$ is RIP of order $q$ and radius $\delta_q$ with
$q\in\{2k, s+2k\}$. It is proved in \cite{candes2008rip} as a result
of the polarization identity, that, for two vectors $u$ and $v$ of
disjoint supports and of sparsity $l$ and $l'$ respectively, if
$\Phi$ is RIP of order $l+l'$, then $|\scp{\Phi u}{\Phi v}|\leq
\delta_{l+l'}\|u\|_2\|v\|_2$. In addition, since $x^*$ is solution of
\iBPDN\ and $x$ is a feasible point of its fidelity constraint,
$\|\Phi h\|_2\leq \|\Phi x^* - y\|_2 + \|y - \Phi x\|_2 \leq
2\epsilon$. Therefore, combining all these considerations,
\begin{multline*}
  \|\Phi h_{T_{|01}}\|_2^2\ \leq\
  2\sqrt{1+\delta_{s+2k}}\,\epsilon\,\|h_{T_{|01}}\|_2 + \sum_{j\geq 2}
  |\scp{\Phi h_{T_{|0}} + \Phi h_{T_1}}{\Phi h_{T_j}}|\\
  \leq\ 2\sqrt{1+\delta_{s+2k}}\,\epsilon\,\|h_{T_{|01}}\|_2 + \big(\delta_{s+2k}\|h_{T_{|0}}\|_2 +
  \delta_{2k}\|h_{T_1}\|_2\big)\,{\sum_{j\geq
      2}}\|h_{T_j}\|_2\\
  \leq 2\sqrt{1+\delta_{s+2k}}\,\epsilon\,\|h_{T_{|01}}\|_2 +\ \mu_{s,k}\,\|h_{T_{|01}}\|_2\,{\sum_{j\geq
      2}}\|h_{T_j}\|_2,\\[-7mm]
\end{multline*}
with $\mu_{s,k}=\sqrt{\delta_{s+2k}^2 + \delta_{2k}^2}$.

Since $(1-\delta_{s+2k})\|h_{T_{|01}}\|_2^2\leq \|\Phi
h_{T_{|01}}\|_2^2$, simplifying the last expression and using
(\ref{eq:bound-on-sum-j-gt-2}) lead to 
$$
(1-\delta_{s+2k})\,\|h_{T_{|01}}\|_2  \leq\ 2\sqrt{1+\delta_{s+2k}}\,\epsilon\ +\
\mu_{s,k}\,\big(2e_0(r;k) + \|h_{T_0}\|_2\big),
$$
or, since $\|h_{T_0}\|_2\leq\|h_{T_{|01}}\|_2$,
$$
\|h_{T_{|01}}\|_2\ \leq\ \alpha\epsilon\ +\
\beta e_0(r;k),
$$
with
$\alpha= 2\sqrt{1+\delta_{s+2k}}\,/\,(1-\delta_{s+2k}-\mu_{s,k})$ and
$\beta={2\mu_{s,k}}\,/\,(1-\delta_{s+2k}-\mu_{s,k})$.

Finally, using again (\ref{eq:bound-on-sum-j-gt-2}),
$$
  \|h\|_2\ \leq\ \|h_{T_{|01}}\|_2 + \|h_{T_{|01}^c}\|_2\ \leq\
  \alpha\epsilon\ +\ (\beta + 2) e_0(r;k)\ +\ \|h_{T_{0}}\|_2\ \leq \
  C_{s,k}\,\epsilon\ +\ D_{s,k}\, e_0(r;k),
$$
with 
$$
C_{s,k} = \frac{4\sqrt{1+\delta_{s+2k}}}{1-\delta_{s+2k}-\mu_{s,k}}, 
$$
and
$$
D_{s,k} = 2\,\frac{1 + \mu_{s,k} - \delta_{s+2k}}{1-\delta_{s+2k}-\mu_{s,k}}.
$$
The denominator of these two constants makes sense only if $1 -
\delta_{s+2k} - \mu_{s,k} > 0$, \ie if $\delta_{2k}^2 + 2
\delta_{s+2k}~<~1$, which provides the announced reconstruction
condition.
\end{proof}

\section{Observations}
\label{sec:obs}

Some observations may be realized from Theorem
\ref{thm:l2l1-inst}. First, in the case where there is no knowledge
about the signal support, \ie $T=\emptyset$ and $s=0$, we do find the
previous sufficient condition of \cite{candes2008rip} characterizing
when BPDN satisfies the $\ell_2-\ell_1$ instance optimality, namely
$\delta_{2k}<\sqrt{2}-1$ as involved by $\delta_{2k}^2 + 2\delta_{2k}
< 1$.

Second, the condition $\delta_{2k}^2 + 2\delta_{s+2k} < 1$ is
satisfied if $\delta_{s+2k}<\sqrt{2}-1$ since we have always
$\delta_{2k}<\delta_{s+2k}$. This seems again a simple generalization
of the previous result in \cite{candes2008rip}, \ie \iBPDN\ is stable
if the RIP of $\Phi$ is guaranteed over the sparsity order $s+2k$ with
a radius $\delta_{s+2k}<\sqrt{2}-1$. Intuitively, the matrix must be
sufficiently ``well conditionned'' to estimate both the unknown values
of $x$ on the known set $T$ and the $k$ other significant values of
$x$ somewhere outside of $T$. This induces somehow the required $s+2k$
RIP sparsity order, where $s$ and $2k$ stand for the degrees of
freedom of $x$ on $T$ and on $T^c$ respectively.

Third, if the signal $x$ is exactly sparse, there is a $k<N-s$ such
that $k=\#\big((\supp x)\setminus T\big)$ and $e_0(r;k)=0$. Without
noise on the measurements, the previous theorem guarantees therefore
the perfect reconstruction of the signal, \ie $x^*=x$, as obtained in
\cite{vaswani5066modified}.

Finally, the compressibility of the signal $x$ is quantified by the
compressibility error $e_0(r,k)$. In other words, the compressibility
is measured from $r=x-x_T$ outside of the known support part $T$ of
$x$. This new measure is of course the simple generalization of the
previous term $e_0(k)=k^{-1/2}\|x-x^k\|_1=e_0(x;k)$ introduced for
instance in \cite{candes2008rip}.

\section{Connection to $\delta$-stable Embeddings and
the Cancel-then-Recover strategy}
\label{sec:related-work}

Theorem \ref{thm:l2l1-inst} has an interesting connection with the
recent work of Davenport et
al.~\cite{Davenport_M_2010_article_sigproccompress} showing that
several signal processing tasks, \ie signal detection, classification,
estimation and filtering, can be realized efficiently on the
compressive measurements of a signal without reconstructing it. In
their work, the Authors study in particular the possibility to
subtract from these measurements the influence of the known part of
the signal support. Let us briefly explain that work before to compare
our work with this of \cite{Davenport_M_2010_article_sigproccompress}.

For this explanation, we use the framework of Section \ref{sec:fmwk}
with the simplifying canonical basis $\Psi = \Id$ and the pure sensing
model $y=\Phi x$. We define also the subspace
$\Sigma_T=\{u\in\Rbb^n:\supp u\subset T\}$ and the matrix $\Omega =
\Phi_T\in \Rbb^{m\times s}$, \ie the restriction of $\Phi$ to the
columns indexed in $T\subset \cl N$. Two operators can be built from
$\Omega$ and its Moore-Penrose pseudoinverse
$\Omega^\dagger=(\Omega^T\Omega)^{-1}\Omega^T$, \ie $P_{\Omega}\ =\
\Omega\Omega^{\dagger}$ and $P_{\Omega^\perp} = 1\ -\
\Omega\Omega^{\dagger}$, the orthogonal projectors on the range of
$\Omega$ and on the nullspace of $\Omega^T$ respectively.

Writing $x=x_T+x_{T^c}$, we can notice that $P_{\Omega^\perp}\Phi
x=P_{\Omega^\perp}\Phi x_{T^c}$. In short, the influence (or
interference) of $x_T$ on $y=\Phi x$ may be canceled without
reconstructing $x$. The idea of the \emph{cancel-then-recover}
strategy promoted in \cite{Davenport_M_2010_article_sigproccompress}
is therefore to reconstruct actually $x_{T^c}$ from $\wt y = \wt
\Phi x = \wt \Phi x_{T^c}$, with $\wt \Phi =
P_{\Omega^\perp}\Phi $. This can be done for instance by solving
either the Basis Pursuit program
$$
  \wt x\ =\ \argmin_u \norm{u}_1\ {\rm s.t.}\ \wt y = \wt \Phi u,
$$
or an equivalent greedy method as CoSaMP
\cite{cosamp,Davenport_M_2010_article_sigproccompress}. Of course,
$\wt x_{T}=0$ since this part of $\wt x$ does not contribute to the
fidelity constraint. It is equivalent to say that the reconstruction
runs over the space $P_{\Id_{T}^\perp}\Rbb^n$, where
$P_{\Id_{T}^\perp} u = u_{T^c}$ for any $u\in\Rbb^n$. Therefore, the
estimation error between $\wt x$ and $x$ can be bounded over $T^c$.

For this purpose $\wt \Phi$ must be characterized in function of
$\Phi$. This can be done by considering a generalization the
Restricted Isometry Property: Given $\delta\in (0,1)$ and two spaces
$\cl U,\cl V\subset\Rbb^n$, a matrix $\Phi$ realizes a
$\delta$-\emph{stable embedding} of $(\cl U,\cl V)$ if
$$
(1-\delta)\,\|u-v\|^2_2\ \leq\ \|\Phi u - \Phi v\|_2^2\ \leq\ (1+\delta)\,\|u-v\|^2_2,
$$
for all $u\in\cl U$ and $v\in \cl V$. In particular the RIP of order
$q$ and radius $\delta_q$ is equivalent to a $\delta_q$-stable
embedding of $(\Sigma_q,\{0\})$, with
$\Sigma_q=\{u\in\Rbb^n:\,\|u\|_0\leq q\}$ the set of $q$-sparse
signals. The following result provides then the desired characterization.
\begin{lemma}[Corollary 4 in \cite{Davenport_M_2010_article_sigproccompress}]
\label{lem:conn-delta-stable}
  Suppose that $\Phi\in\Rbb^{m\times n}$ is a $\delta$-stable
  embedding of $(\Sigma_{2k},\Sigma_T)$. Then $\wt \Phi$ is a
  $\delta/(1-\delta)$-stable embedding of
  $(P_{\Id_{T}^\perp}\Sigma_{2k},\{0\})$.     
\end{lemma}
In particular, this Lemma implies that if $\Phi$ is RIP of order
$s+2k$ with radius $\delta_{s+2k}$, it is then a
$\delta_{s+2k}$-stable embedding of $(\Sigma_{2k},\Sigma_T)$, and
therefore, $\wt \Phi$ is RIP of order $2k$ and radius
$\delta_{s+2k}/(1-\delta_{s+2k})$ over the space
$P_{\Id_{T}^\perp}\Rbb^n\simeq \Rbb^{n-s}$. The $\ell_2-\ell_1$
instance optimality of the BP program \cite{candes2008rip} above holds
if $\delta'=\delta_{s+2k}/(1-\delta_{s+2k})<\sqrt{2}-1$, \ie if
$\delta_{s+2k}<(\sqrt{2}-1)/\sqrt{2}$. In that case,
\begin{equation}
\label{eq:ctr-bound}
\|x_{T^c}-\wt x_{T^c}\|_2\ \leq\ \wt D_{\delta'}\, e_0(x_{T^c},k)\ =\
\wt D_{\delta'}\, e_0(r,k),
\end{equation}
with $\wt D_{\delta'} = 2\,\frac{1 + (\sqrt 2 - 1)\delta'}{1-(\sqrt
  2 + 1)\delta'} = 2\,\frac{1 + (\sqrt 2 - 2)\delta_{s+2k}}{1-(\sqrt
  2 + 2)\delta_{s+2k}}$.

In this paper, we show that \iBPDN\ is optimal when $\delta^2_{2k} +
2\delta_{s+2k} < 1$. This condition is weaker than the one proposed in
\cite{Davenport_M_2010_article_sigproccompress},
i.e. $\delta_{s+2k}<(\sqrt{2}-1)/\sqrt{2}$, however it is interesting
to notice that both consider also the RIP of order $2s+k$ and both
are stable for compressible signals.  Moreover, \iBPDN\ gives
guarantees for the estimation of the whole signal and not only for its
behavior over $T^c$. Of course, if $x^*$ is the solution of \iBPDN\
(with $\epsilon=0$), we get similarly
$$
\|x_{T^c}- x^*_{T^c}\|_2\ \leq\ \|x - x^*\|_2\leq\ D_{s,k}\, e_0(r,k),
$$
with $D_{s,k} < 2\,\frac{1 + (\sqrt 2 - 1)\delta_{s+2k}}{1-(\sqrt 2 +
  1)\delta_{s+2k}} < \wt D_{\delta'}$.

We can remark also that, conversely to the current cancel-then-recover
strategy\footnote{Robustness of this strategy against an additional
  noise $n$ could be obtained by bounding the power of
  $P_{\Omega^\perp}n$ when $y=\Phi x + n$.}, \iBPDN\ provides
stability against noisy measurements. An open question is however that
$\Phi$ in \cite{Davenport_M_2010_article_sigproccompress} has not to
be really RIP of order $s+2k$ to valid (\ref{eq:ctr-bound}). As
reported in Lemma \ref{lem:conn-delta-stable}, $\Phi$ simply needs to
provide a $\delta$-stable embedding over $(\Sigma_{2k},\Sigma_T)$
which is weaker than asking the RIP of order $s+2k$. Given $k$ and
$m$, that second requirement holds possibly for a smaller radius
$\delta$ than the RIP radius $\delta_{s+2k}$.

\section{Numerical Method}
\label{sec:algo}

In this section, we sketch of a simple algorithm for the reader
interested in a numerical implementation of \iBPDN. This one relies on
monotone operator splitting and proximal methods
\cite{combettes2007drs,Fadili2009}.  At the heart of this procedure is
the definition of the \emph{proximity operator} of any convex function
$\varphi:\Rbb^n\to\Rbb$, \ie the unique solution of $\prox_\varphi(z)
= \arg\min_{u}\inv{2}\|u - x\|_2^2 + \varphi(z)$.

Both BPDN and \iBPDN\ are special cases of the general minimization
problem
\begin{equation*}
  \label{eq:convex-prob}
  \arg\min_{x \in \Hm}\ f(x) + g(x).\eqno{({\rm \bf P})}
\end{equation*}
For \iBPDN, $f(u)=\|u_{T^c}\|_1$ and $g(u) = \imath_{C(\epsilon)}(u)=
0$ if $u\in C(\epsilon)$ and $\infty$ otherwise, \ie the
\emph{indicator function} of the closed convex set
$C(\epsilon)=\{v\in\Rbb^n:\|y - \Phi v\|_2\leq \epsilon\}$.

Of course $f$ and $g$ are both non-differentiable, however, since (i)
their domain is non-empty, (ii) they are convex and (iii) lower
semi-continuous (lsc), \ie $\liminf_{u \to u_0} f(u) = f(u_0)$ for all
$u_0 \in \dom f$, \iBPDN\ can be solved by the following
Douglas-Rachford iterative method \cite{Fadili2009}:
\begin{equation}
\label{eq:DR-iter}
u^{(t+1)} = (1-\tfrac{\alpha_t}{2})\,u^{(t)} + 
\tfrac{\alpha_t}{2}\,S^\odot_{\gamma}\circ\Proj^\odot_{C(\epsilon)}(u^{(t)}),
\end{equation}
where $A^\odot \triangleq 2A - \Id$ for any operator $A$, $\alpha_t
\in (0,2)$ for all $t \in \mathbb{N}$, $S_{\gamma}=\prox_{\gamma f}$
for some $\gamma>0$ and $\Proj_{C(\epsilon)} = \prox_{g}$ is the
orthogonal projection onto the tube $C(\epsilon)$. From
\cite{combettes2004smi}, one can show that the sequence $(u^{(t)})_{t
  \in \mathbb{N}}$ converges to some point $u^*$ and $x^* =
\Proj_{C(\epsilon)}(u^*)$ is the solution of \iBPDN.

We may compute that $S_{\gamma} z=\prox_{\gamma f} z$ is actually the
component-wise soft-thresholding operator of $z$ on $T^c$,
\ie $(S_{\gamma} z)_i = \sign z_i (|z_i| - \gamma)_+$ if $i\in T^c$
and $z_i$ if $i\in T$, with, for $\lambda\in\Rbb$, $(\lambda)_+ =
\lambda$ if $\lambda\geq 0$ and 0 else.  Efficient ways to compute
$\Proj_{C(\epsilon)}$ are also given in \cite{Fadili2009}.

\section{Conclusion}
\label{sec:conclusion}

This short note has studied the modification of Compressed Sensing
introduced in \cite{vaswani5066modified,Vaswani2009}, \ie when the
signal sparsity assumption is increased by the knowledge of a part of
its support. We showed theoretically that a simple generalization of
the common Basis Pursuit DeNoise program, \ie the \emph{innovative}
BPDN, has similar stability guarantees than BPDN with respect to both
signal compressibility and noisy measurements. Interestingly, the
obtained requirements are related to the conclusion of
\cite{Davenport_M_2010_article_sigproccompress} when the
cancel-then-recover strategy is applied to the context of this paper.

In the future, we plan to investigate possible numerical applications
of this formalism. In particular, when \iBPDN\ is integrated to the
reconstruction of sequences of sparse or compressible signals, we
would like to assess the quality of the reconstruction in function of
the number of measurements when the amount of innovation, \ie the
ratio between the unknown and the known signal support parts, can be
quantified over time.

\section{Acknowledgements}

We are very grateful to Prof. Pierre Vandergheynst (Signal Processing
Laboratory, LTS2/EPFL, Switzerland) for his useful advices and his
hospitality during their postdoctoral stay in EPFL.

\end{document}